\documentclass{article}
\usepackage[dvips,letterpaper,margin=1.0in]{geometry}
\usepackage{sty}

\usepackage{authblk}
\author{Alexander S.\ Wein\thanks{Email: \textit{awein@cims.nyu.edu}. Partially supported by NSF grant DMS-1712730 and by the Simons Collaboration on Algorithms and Geometry. Part of this work was done while the author was visiting the Simons Institute for the Theory of Computing.}}

\affil{Department of Mathematics, Courant Institute of Mathematical Sciences, NYU}

\date{}

\title{Optimal Low-Degree Hardness of Maximum Independent Set}

\begin{document}

\maketitle

\begin{abstract}
We study the algorithmic task of finding a large independent set in a sparse \ER random graph with $n$ vertices and average degree $d$. The maximum independent set is known to have size $(2 \log d / d)n$ in the double limit $n \to \infty$ followed by $d \to \infty$, but the best known polynomial-time algorithms can only find an independent set of half-optimal size $(\log d / d)n$. We show that the class of \emph{low-degree polynomial algorithms} can find independent sets of half-optimal size but no larger, improving upon a result of Gamarnik, Jagannath, and the author. This generalizes earlier work by Rahman and Vir\'ag, which proved the analogous result for the weaker class of \emph{local algorithms}.
\end{abstract}

\section{Introduction}

We consider the problem of finding a large independent set (i.e., a set of vertices such that no two are adjacent) in the sparse \ER graph $G(n,d/n)$ where each of the $\binom{n}{2}$ potential edges on vertex set $[n]$ occurs independently with probability $d/n$. In the double limit $n \to \infty$ followed by $d \to \infty$, the largest independent set $S_{\max}$ is known to have asymptotic size $(2 \log d/d)n$. More precisely, as $n \to \infty$ with $d > 0$ fixed we have $\frac{1}{n} |S_{\max}| \to \alpha_d$ with high probability, for some $\alpha_d$ satisfying $\alpha_d = (1+o_d(1)) (2 \log d / d)$ as $d \to \infty$~\cite{frieze-indep,BGT-interp}. We will be interested in the associated algorithmic task: give a polynomial-time algorithm that takes as input a graph drawn from $G(n,d/n)$ and outputs (with high probability) a large independent set. We assume $d$ is known to the algorithm, although it can be estimated easily from the total number of edges. The influential work of Karp~\cite{karp} showed that a simple greedy algorithm can find an independent set of asymptotic size $(\log d/d)n$, which is half of the optimum. Decades later, we still do not know a polynomial-time algorithm to find an independent set of size $(1+\eps)(\log d/d)n$ for any fixed $\eps > 0$ (independent of both $d$ and $n$). Moreover, evidence has emerged to suggest that no such algorithm exists. It was shown by Coja-Oghlan and Efthymiou~\cite{CE-indep} (building on~\cite{AC-barriers}) that the independent sets of size larger than half-optimal are ``clustered'' in a way that implies slow mixing of the Metropolis process for sampling such sets. Furthermore, it was shown by Rahman and Vir\'ag~\cite{half-opt} (building on~\cite{LW-indep,GS-local}) that the class of \emph{local algorithms} can find independent sets of half-optimal size \emph{and no larger}. Here, a local algorithm (also called \emph{i.i.d.\ factors}) allows each vertex to decide whether or not to include itself in the set based only on its local neighborhood in the graph (of constant radius) along with i.i.d.\ random variables attached to the vertices (see Section~\ref{sec:upper} for a formal definition).

The above results suggest that $(\log d/d)n$ may be the fundamental limit for polynomial-time algorithms. In this work we provide further evidence for this by showing that $(\log d/d)n$ is the fundamental limit for the class of \emph{low-degree polynomial algorithms} (to be defined formally in the next section) where each vertex's membership (or non-membership) in the independent set is determined by thresholding a low-degree multivariate polynomial of the edge-indicator variables that describe the input graph. This class of low-degree algorithms includes the class of local algorithms mentioned above (see Remark~\ref{rem:local-to-ld}), and also (as discussed in Appendix~A of~\cite{GJW-lowdeg}) includes other popular algorithmic paradigms such as approximate message passing (e.g.\ \cite{amp,BM-amp,JM-amp,sk-opt,EMS-opt}) and power iteration\footnote{Notably, low-degree algorithms capture power iteration on \emph{any} matrix that is itself low-degree in the input. This allows for non-trivial spectral methods such as the \emph{tensor unfolding} method for tensor PCA~\cite{RM-tensor,sos-tensor-pca}, which outperforms more ``standard'' algorithms such as message passing and gradient descent~\cite{RM-tensor,BGJ-tensor}.}. Furthermore, starting from the influential line of work~\cite{pcal,HS-bayesian,sos-power,sam-thesis}, it has been established that low-degree algorithms (with degree logarithmic in the dimension) are precisely as powerful as the best known polynomial-time algorithms for a number of problems in high-dimensional statistics including planted clique, sparse PCA, community detection, tensor PCA, and many others~\cite{HS-bayesian,sos-power,sam-thesis,sk-cert,lowdeg-notes,subexp-sparse,heavy-tailed,secret-leakage,LZ-tensor,lowdeg-rec,spec-planting,sq-ld}. Thus, failure of low-degree algorithms is a form of concrete evidence for computational hardness of statistical problems. For more on low-degree algorithms, we refer the reader to~\cite{lowdeg-notes} (for a survey on the setting of~\emph{hypothesis testing}), \cite{lowdeg-rec} (for the setting of~\emph{estimation}), or~\cite{GJW-lowdeg} (for the setting of random optimization problems, which is the relevant setting for this work).

Most prior work on low-degree algorithms has focused on problems with a ``planted'' signal, in which case failure of low-degree algorithms can be shown via a direct linear-algebraic computation. This technique does not apply to ``non-planted'' problems such as the maximum independent set problem that we consider here, and so a different approach is needed which leverages structural properties of the solution space (see Section~\ref{sec:tech}). For non-planted problems, the first results for low-degree algorithms were given by Gamarnik, Jagannath, and the author~\cite{GJW-lowdeg} (building on~\cite{GJ-amp}), who showed that low-degree algorithms cannot find independent sets of size exceeding $(1+1/\sqrt{2})(\log d/d)n$ in $G(n,d/n)$. Here we improve this to the optimal threshold $(\log d/d)n$. We also provide the matching positive result, showing that $(\log d/d)n$ is achievable by low-degree algorithms (following a proof sketch given in~\cite{GJW-lowdeg}). This is the first non-planted problem for which matching upper and lower bounds have been obtained on the objective value attainable by low-degree algorithms (apart from trivial cases where the global optimum value can be reached). One conceptual advantage of our results over the existing results for local algorithms is that low-degree algorithms offer a unified framework to explain computational hardness in a wide variety of high-dimensional problems, whereas local algorithms are specific to problems involving sparse graphs. This is exemplified by the fact that our impossibility result can be extended to the case of dense graphs such as $G(n,1/2)$; see Section~\ref{sec:ext}.

\subsection{Main Results}

We now formally define the problem setup, following~\cite{GJW-lowdeg}. We say that a function $f:\RR^m\to\RR^n$ is a \emph{polynomial of degree (at most) $D$} if it may be written in the form 
\begin{equation}\label{eq:vec-val-poly}
f(Y) = (f_1(Y),\ldots,f_n(Y)),
\end{equation}
where each $f_i:\RR^m\to\RR$ is a multivariate polynomial (in the usual sense) of degree at most $D$ with real coefficients. We also define a \emph{random polynomial} $f: \RR^m \to \RR^n$ in the same way but where the coefficients may be random (but independent from the input $Y$): formally, for some probability space $(\Omega,\PP_\omega)$, $f$ is a map $f: \RR^m \times \Omega \to \RR^n$ such that $f(\cdot,\omega)$ is a degree-$D$ polynomial for each ``seed'' $\omega \in \Omega$. (We will see that randomness does not actually help; see Lemma~\ref{lem:rand-det}.)

For our purposes, the input to $f$ will be an $n$-vertex graph encoded as $Y \in \{0,1\}^m$ with $m = \binom{n}{2}$, where each entry of $Y$ is the indicator variable for the presence of a particular edge. We write $Y \sim G(n,d/n)$ for an \ER graph, i.e., $Y$ is i.i.d.\ Bernoulli$(d/n)$.

We need to define what it means for a polynomial $f: \RR^m \to \RR^n$ to find an independent set in a graph $Y$. Instead of asking $f(Y)$ to be the indicator vector of an independent set, we relax this somewhat and ask only for a ``near-indicator vector'' of a ``near-independent set''. More precisely, the following ``rounding'' procedure from~\cite{GJW-lowdeg} will be used to extract an independent set from the output of $f$.

\begin{definition}\label{def:V}
Let $f: \RR^m \to \RR^n$ be a random polynomial with $m = \binom{n}{2}$. For $Y \in \{0,1\}^m$, and $\eta \ge 0$, let $V^\eta_f(Y,\omega)$ be the independent set in the graph $Y$ obtained by the following procedure. Let
\[ A = \{i \in [n] \,:\, f_i(Y,\omega) \ge 1\}, \]
\[ \tilde A = \{i \in A \,:\, \text{$i$ has no neighbors in $A$ in the graph $Y$}\}, \]
and
\[ B = \{i \in [n] \,:\, f_i(Y,\omega) \in (1/2,1)\}. \]
Then define
\begin{equation}\label{eq:V}
V^\eta_f(Y,\omega) = \left\{\begin{array}{ll} \tilde A & \text{if } |A \setminus \tilde A| + |B| \le \eta n, \\ \emptyset & \text{otherwise.} \end{array}\right.
\end{equation}
\end{definition}

\noindent Informally speaking, $f_i$ should output a value $\ge 1$ to indicate that vertex $i$ is in the independent set and should output a value $\le 1/2$ to indicate that it is not. We allow a small number of ``errors'': there can be up to $\eta n$ vertices where either $f_i(Y) \in (1/2,1)$ or the independence constraint is violated. Vertices that violate the independence constraint are thrown out, and if too many errors are made then the output is the empty set $\emptyset$ (which is thought of as a ``failure'' event). While the choice of thresholds $1$ and $1/2$ is somewhat arbitrary, the interval $(1/2,1)$ of disallowed outputs is important for our impossibility result (Theorem~\ref{thm:main-lower}), as this ensures that a small change in $f(Y,\omega)$ cannot induce a large change in the resulting independent set $V_f^\eta(Y,\omega)$ without encountering the failure event $\emptyset$. On the other hand, our achievability result (Theorem~\ref{thm:main-upper}) will give a low-degree polynomial for which most outputs $f_i(Y)$ lie in $\{0,1\}$ exactly, i.e., it succeeds even under the more stringent definitions $A = \{i \in [n] \,:\, f_i(Y,\omega) = 1\}$ and $B = \{i \in [n] \,:\, f_i(Y,\omega) \notin \{0,1\}\}$.

\begin{definition}\label{def:opt}
For parameters $k > 0$, $\delta \in [0,1]$, $\gamma \ge 1$, and $\eta > 0$, a random polynomial $f: \RR^m \to \RR^n$ is said to \emph{$(k,\delta,\gamma,\eta)$-optimize} the independent set problem in $G(n,d/n)$ if the following are satisfied when $Y \sim G(n,d/n)$:
\begin{itemize}
    \item $\displaystyle \EE_{Y,\omega} \left[\|f(Y,\omega)\|^2\right] \le \gamma k$, and
    \item $\displaystyle \prob_{Y,\omega}\left[|V_f^\eta(Y,\omega)| \ge k\right] \ge 1-\delta$.
\end{itemize}
\end{definition}
\noindent Here, $k$ is the size of the independent set that is produced, $\delta$ is the algorithm's failure probability, $\gamma$ is a normalization parameter, and $\eta$ is the error tolerance of the rounding procedure $V_f^\eta$.

We now state our main results. Theorem~\ref{thm:main-lower} shows that no low-degree polynomial can find an independent set of size $(1+\eps)\frac{\log d}{d} n$, while Theorem~\ref{thm:main-upper} shows that some low-degree polynomial can find an independent set of size $(1-\eps)\frac{\log d}{d} n$. The proofs are given in Sections~\ref{sec:lower} and~\ref{sec:upper}, respectively. The results are interpreted in the remarks below.

\begin{theorem}[Impossibility]\label{thm:main-lower}
For any $\eps > 0$ there exists $d^* > 0$ such that for any $d \ge d^*$ there exists $n^* > 0$, $\eta > 0$, $C_1 > 0$, and $C_2 > 0$ (depending on $\eps,d$) such that the following holds. Let $n \ge n^*$, $\gamma \ge 1$, and $1 \le D \le \frac{C_1 n}{\gamma \log n}$, and suppose $\delta \ge 0$ satisfies
\begin{equation}\label{eq:delta}
\delta \le \exp\left(-C_2 \gamma D \log n\right).
\end{equation}
Then for $k = (1+\eps) \frac{\log d}{d} n$, there is no random degree-$D$ polynomial that $(k,\delta,\gamma,\eta)$-optimizes the independent set problem in $G(n,d/n)$.
\end{theorem}

\begin{theorem}[Achievability]\label{thm:main-upper}
For any $\eps > 0$ there exists $d^* > 0$ such that for any $d \ge d^*$ and any $\eta > 0$ there exists $n^* > 0$, $D > 0$, $\gamma \ge 1$, and $C > 0$ (depending on $\eps,d,\eta$) such that the following holds for all $n \ge n^*$. For $k = (1-\eps)\frac{\log d}{d} n$ and $\delta = \exp(-C n^{1/3})$, there exists a (deterministic) degree-$D$ polynomial that $(k,\delta,\gamma,\eta)$-optimizes the independent set problem in $G(n,d/n)$.
\end{theorem}

\noindent A number of remarks are in order.

\begin{remark}
The results are non-asymptotic but can be thought of as capturing the double limit $n \to \infty$ followed by $d \to \infty$. In other words, $d$ is a large constant depending on $\eps$, and $n$ must then be chosen sufficiently large (where ``sufficiently large'' depends on $d$). In the sequel, asymptotic notation such as $O(\cdot)$ pertains to the limit $n \to \infty$ with all other parameters fixed; parameters not depending on $n$ are considered ``constants''.
\end{remark}

\begin{remark}
The ``tolerance'' parameter $\eta$ should be thought of as a small constant. The impossibility result shows that \emph{some} $\eta > 0$ (depending on $\eps,d$) is \emph{not} achievable, whereas the achievability result show that \emph{any} $\eta > 0$ is achievable. The ``normalization'' parameter $\gamma$ should be thought of as a large constant. The impossibility result shows that \emph{any} $\gamma \ge 1$ is not achievable, whereas the achievability result shows that \emph{some} $\gamma \ge 1$ (depending on $\eps,d,\eta$) is achievable.
\end{remark}

\begin{remark}
Typically, when proving impossibility results for low-degree algorithms, the goal is to rule out any degree $D = O(\log n)$ because polynomials of this degree can capture the best known algorithms for a wide array of problems. In our case, a constant degree $D$ (depending on $\eps,d,\eta$) is sufficient for the achievability result. On the other hand, our impossibility result rules out a much wider range of $D$ values: $D \lesssim n/\log n$. However, the requirement~\eqref{eq:delta} gives an additional tradeoff between $D$ and the failure probability $\delta$. This is present for technical reasons, and ideally we would replace~\eqref{eq:delta} by a milder condition such as $\delta = o(1)$. Still, note that the parameters $D = O(1)$ and $\delta = \exp(-\Omega(n^{1/3}))$ in our achievability result lie well within the set of $(D,\delta)$ pairs ruled out by our impossibility result.
\end{remark}

\begin{remark}
In the achievability result, the value of $\delta$ can likely be improved from $\exp(-Cn^{1/3})$ to $\exp(-Cn)$. This can perhaps be accomplished by using the powerful machinery of~\cite{empirical-nbhd} in place of Corollary~\ref{cor:local-conc}, but we do not attempt this here.
\end{remark}

\subsection{Proof Techniques}
\label{sec:tech}

We now give an overview of the proof techniques and discuss their relation to prior work. We first discuss the achievability result (Theorem~\ref{thm:main-upper}). It is known already that local algorithms can find independent sets of size $(1-\eps)\frac{\log d}{d} n$~\cite{half-opt}. Our proof transfers this to a result about low-degree algorithms by showing that \emph{any} local algorithm can be well-approximated by a constant-degree polynomial. A proof sketch of this reduction was given already in Appendix~A of~\cite{GJW-lowdeg}, but here we give the full details and determine the values of the parameters $D, \delta, \gamma, \eta$. The main difficulty lies in establishing that the failure probability $\delta$ is very small; for this we appeal to a result of~\cite{very-sparse} that gives tail bounds for certain ``local'' functions on sparse random graphs.

We now discuss the impossibility result (Theorem~\ref{thm:main-lower}), which is our main contribution. This result falls into a line of work initiated by Gamarnik and Sudan~\cite{GS-local}, who showed that local algorithms fail to find independent sets larger than $(1+1/\sqrt{2}) \frac{\log d}{d} n$. Their proof harnessed the so-called \emph{overlap gap property (OGP)}: in a typical graph drawn from $G(n,d/n)$, there are no two independent sets that each have size exceeding $(1+1/\sqrt{2}) \frac{\log d}{d} n$ and have intersection size (``overlap'') roughly $\frac{\log d}{d} n$. They used an interpolation argument to show that if a hypothetical local algorithm were to succeed at finding independent sets larger than $(1+1/\sqrt{2}) \frac{\log d}{d} n$, this could be used to construct two independent sets violating the OGP, leading to a contradiction. This proof technique was subsequently extended in two important ways. First, Rahman and Vir\'ag~\cite{half-opt} improved the threshold for failure of local algorithms down to $\frac{\log d}{d} n$, which is optimal. The proof involves establishing a more intricate ``forbidden'' structure that involves many independent sets with a particular intersection pattern (in contrast to the OGP, which involves only two sets). Again, a hypothetical local algorithm can be used to construct this forbidden structure, leading to a contradiction. This idea inspired further work in the area of random constraint satisfaction problems~\cite{nae-sat,walksat}. A separate line of work~\cite{max-cut-local,GJ-amp,GJW-lowdeg} extended the ideas of Gamarnik and Sudan~\cite{GS-local} in a different direction: instead of the basic OGP discussed above, they consider an ``ensemble'' variant of OGP in which a particular overlap between two large independent sets is forbidden even when the independent sets do not come from the same graph but from two correlated random graphs. This variant of OGP can be used not only to rule out local algorithms, but also to rule out any sufficiently ``stable'' algorithm (which roughly means that a small change to the input only causes a small change to the output); this idea was first discovered by~\cite{GJ-amp} and was later used by~\cite{GJW-lowdeg} to show that low-degree algorithms---which are stable---cannot find independent sets larger than $(1+1/\sqrt{2}) \frac{\log d}{d} n$.

To prove our impossibility result, we combine the two main ideas discussed above: we consider a forbidden structure that involves many independent sets and also involves many correlated random graphs. The crux of the proof lies in the specific choice of this forbidden structure (see Proposition~\ref{prop:forbid}), which is carefully chosen so that (i) with high probability, no instance of this structure occurs, and (ii) a hypothetical stable algorithm can be used to construct an instance of this structure, leading to a contradiction. On a technical level, our forbidden structure is quite different from the one used by Rahman and Vir\'ag~\cite{half-opt} in that theirs is highly symmetric, e.g., any two of the sets involved have the same intersection size. This is suitable for their purposes because due to special properties of local algorithms, a hypothetical local algorithm can be used to construct such a symmetric collection of sets. In our case, however, it is not clear that a hypothetical low-degree algorithm can be used to construct a symmetric collection of sets; we instead define a new class of forbidden structures that are not necessarily symmetric. Finally, we remark that the only property of low-degree polynomials that we use is their ``stability'' (in the sense of Proposition~\ref{prop:stable}), and so the proof actually rules out all ``stable'' algorithms.

\subsection{Extensions and Future Directions}
\label{sec:ext}

In this work we have given the first techniques for obtaining sharp impossibility results for low-degree algorithms in random optimization problems (with no planted signal). Hopefully these techniques can be adapted to other non-planted settings such as random constraint satisfaction problems (e.g.~\cite{AC-barriers,nae-sat,walksat}) and spin glass optimization problems~\cite{subag,sk-opt,GJ-amp,EMS-opt,GJW-lowdeg}. Low-degree algorithms are a promising candidate for a unified framework to explain computational hardness in a wide array of non-planted problems, analogous to the more established low-degree framework for planted problems.

One possible extension of our results is to consider the same independent set problem but in denser graphs. For instance, in $G(n,1/2)$ the largest independent set has size $2 \log_2 n$, but the best known polynomial-time algorithm is a simple greedy algorithm which can find an independent set of half-optimal size $\log_2 n$~\cite{karp}. An argument nearly identical to the proof of Theorem~\ref{thm:main-lower} yields the following result which shows that low-degree algorithms cannot improve upon this.

\begin{theorem}
For any $\eps > 0$ there exists $n^* > 0$, $\tilde\eta > 0$, $C_1 > 0$, and $C_2 > 0$ (depending on $\eps$) such that the following holds. Let $n \ge n^*$, $\gamma \ge 1$, and $1 \le D \le \frac{C_1 \log^2 n}{\gamma}$, and suppose $\delta \ge 0$ satisfies
\[ \delta \le \exp(-C_2 \gamma D - 2 \log n). \]
Then for $k = (1+\eps) \log_2 n$ and $\eta = \tilde\eta \log_2 n / n$, there is no random degree-$D$ polynomial that $(k,\delta,\gamma,\eta)$-optimizes the independent set problem in $G(n,1/2)$.
\end{theorem}

\noindent However, the matching achievability result remains open: it is not clear how to write the greedy algorithm as a low-degree polynomial or otherwise give a low-degree algorithm that finds an independent set of size $(1-\eps) \log_2 n$. We expect that it should be possible to obtain such a low-degree algorithm (perhaps of degree $D = O(\log n)$ and failure probability $\delta = \exp(-\Omega(\log^2 n))$) via the approximate message passing framework, which has been successful in other non-planted settings~\cite{nonneg,sk-opt,EMS-opt,perceptron}.

\subsection*{Notation}

Asymptotic notation such as $o(1)$ or $\Omega(n)$ pertains to the limit $n \to \infty$ with all other parameters (such as $d$) held fixed; in other words, parameters not depending on $n$ are considered ``constants'' and may be hidden by this notation. On the other hand, $o_d(1)$ denotes a quantity that depends on $d$ but not $n$, and tends to $0$ as $d \to \infty$ (with all other parameters held fixed).

Throughout, we will use the shorthand $m = \binom{n}{2}$ and $\Phi = \frac{\log d}{d} n$. We define $[n] = \{1,2,\ldots,n\}$ and use $\|\cdot\|$ for the $\ell^2$-norm of a vector. All logarithms use the natural base unless stated otherwise. All graphs are assumed to have no self-loops nor parallel edges.

\section{Proof of Impossibility}\label{sec:lower}

In this section we prove our main impossibility result (Theorem~\ref{thm:main-lower}) which shows that low-degree algorithms cannot find independent sets of size $(1+\eps) \frac{\log d}{d} n$.

\subsection{Interpolation Path}

Here we define a sequence of correlated random graphs that will be central to the argument. We will represent a graph on vertex set $[n]$ by $Y \in \{0,1\}^m$ where $m = \binom{n}{2}$. Here $Y_1,\ldots,Y_m$ are indicator variables for the edges (where $0$ indicates a non-edge and $1$ indicates an edge), listed in some fixed but arbitrary order.

\begin{definition}\label{def:interp}
For $T \in \NN$, consider the \emph{length-$T$ interpolation path} $Y^{(0)},\ldots,Y^{(T)}$ sampled as follows. First, $Y^{(0)} \sim G(n,d/n)$. Then for $1 \le t \le T$, $Y^{(t)}$ is obtained from $Y^{(t-1)}$ by resampling coordinate $\sigma(t) \in [m]$ from $\mathrm{Bernoulli}(d/n)$. Here $\sigma(t) = t - k_t m$ where $k_t$ is the unique integer for which $1 \le \sigma(t) \le m$.
\end{definition}

\subsection{Forbidden Structures}

The proof will hinge on the non-existence of certain structures (primarily the one defined in Proposition~\ref{prop:forbid}) with high probability over the interpolation path. The following standard bounds will be used repeatedly:
\begin{equation}\label{eq:binom-bound}
\binom{n}{k} \le \left(\frac{en}{k}\right)^k \qquad \text{for all integers } 1 \le k \le n,
\end{equation}
\begin{equation}\label{eq:log-bound}
(1-x)^r \le \exp(-rx) \qquad \text{for all } x \in \RR,\, r > 0. 
\end{equation}

\noindent We start with a well-known upper bound on the size of the maximum independent set in $G(n,d/n)$.

\begin{lemma}\label{lem:max-size}
Fix $\eps > 0$. If $d > 0$ is a sufficiently large constant (depending on $\eps$), then with probability $1-\exp(-\Omega(n))$ there is no independent set in $G(n,d/n)$ of size exceeding $(2+\eps) \frac{\log d}{d}n$.
\end{lemma}
\begin{proof}
Let $\Phi = \frac{\log d}{d} n$ and define $a \ge 2+\eps$ so that $a\Phi = \lceil (2+\eps)\Phi \rceil$. Let $N$ denote the number of independent sets of size exactly $a\Phi$; our goal is to show $N = 0$ with high probability. The proof will use a simple first moment method: we compute $\EE[N]$ and show that it is exponentially small. We have
\begin{align*}
\EE[N] &= \binom{n}{a \Phi} (1-d/n)^{\binom{a\Phi}{2}} \\
&\le \left(\frac{en}{a\Phi}\right)^{a\Phi} \exp\left(-\frac{d}{n}\binom{a\Phi}{2}\right) & &\text{using \eqref{eq:binom-bound} and \eqref{eq:log-bound}} \\
&= \exp\left[a \Phi \log\left(\frac{ed}{a \log d}\right) - \frac{d a^2 \Phi^2}{2n} + O(1)\right] \\
&= \exp\left[\Phi \log d \left(a - \frac{a^2}{2} + o(1) + o_d(1)\right)\right] \\
&\le \exp\left[\Phi \log d \left(-\eps + o(1) + o_d(1)\right)\right] & &\text{using } a \ge 2+\eps \\
&= \exp(-\Omega(n))
\end{align*}
for sufficiently large $d$. The result follows by Markov's inequality.
\end{proof}

\noindent The forbidden structure defined in the following result will be the crux of the proof.

\begin{proposition}\label{prop:forbid}
Fix constants $\eps > 0$ and $K \in \NN$ with $K \ge 1 + 5/\eps^2$. Consider the interpolation path $Y^{(0)},\ldots,Y^{(T)}$ from Definition~\ref{def:interp}, of any length $T = n^{O(1)}$. If $d > 0$ is a sufficiently large constant (depending on $\eps, K$), then with probability $1-\exp(-\Omega(n))$ there does not exist a sequence of sets $S_1,\ldots,S_K \subseteq [n]$ satisfying the following properties:
\begin{enumerate}
    \item[(i)] for each $k \in [K]$ there exists $0 \le t_k \le T$ such that $S_k$ is an independent set in $Y^{(t_k)}$,
    \item[(ii)] $|S_k| \ge (1+\eps) \frac{\log d}{d} n$ for all $k \in [K]$,
    \item[(iii)] and $|S_k \setminus (\cup_{\ell < k}\, S_\ell)| \in \left[\frac{\eps}{4} \frac{\log d}{d} n,\frac{\eps}{2} \frac{\log d}{d} n\right]$ for all $2 \le k \le K$.
\end{enumerate}
\end{proposition}

\begin{proof}
Let $N$ denote the number of sequences $(S_1,\ldots,S_K)$ satisfying the properties (i)-(iii). The proof will use the first moment method: we compute $\EE[N]$ and show that it is exponentially small. Let $\Phi = \frac{\log d}{d} n$. Let $a_k$ and $b_k$ be defined by $|S_k| = a_k \Phi$ and $|S_k \setminus (\cup_{\ell < k}\, S_\ell)| = b_k \Phi$, and note that (ii) and (iii) state that $a_k \ge 1+\eps$ and $b_k \in \left[\frac{\eps}{4},\frac{\eps}{2}\right]$. Also let $c$ be defined by $|\cup_k S_k| = c \Phi$, and note that (iii) implies $c \le a_1 + (K-1) \frac{\eps}{2}$. By Lemma~\ref{lem:max-size}, we can assume $a_k \le 2+\eps$. Thus, $c$ is upper-bounded by a constant $C(\eps,K)$ that does not depend on $d$. We need to count the number of sequences $(S_1,\ldots,S_K)$. There are at most $n^{2K}$ choices for the values $\{a_k\}$ and $\{b_k\}$. Once $\{a_k\}$ and $\{b_k\}$ are fixed, the number of ways to choose $\{S_k\}$ is at most
\begin{align*}
\binom{n}{a_1 \Phi} \prod_{k=2}^K \binom{n}{b_k \Phi}&\binom{c\Phi}{(a_k-b_k)\Phi} \le \left(\frac{en}{a_1 \Phi}\right)^{a_1\Phi} \prod_{k=2}^K \left(\frac{en}{b_k \Phi}\right)^{b_k\Phi} \left(\frac{ec}{a_k-b_k}\right)^{(a_k-b_k)\Phi} \qquad\qquad \text{using \eqref{eq:binom-bound}} \\
&= \exp\left\{a_1 \Phi \log\left(\frac{ed}{a_1 \log d}\right) + \sum_{k=2}^K \left[ b_k\Phi \log\left(\frac{ed}{b_k \log d}\right) + (a_k-b_k)\Phi \log\left(\frac{ec}{a_k-b_k}\right)\right]\right\} \\
&= \exp\left\{\Phi \log d \left(a_1 + \sum_{k=2}^K b_k + o_d(1)\right)\right\}
\end{align*}
where we have used $a_k \in [1 + \eps, 2 + \eps]$, $b_k \in \left[\frac{\eps}{4},\frac{\eps}{2}\right]$, and $c \le C(\eps,K)$ to conclude that certain terms are $o_d(1)$.

Now for a fixed $\{S_k\}$ satisfying (ii) and (iii), we need to upper-bound the probability that (i) is satisfied. We will take a union bound over the possible choices of $\{t_k\}$ in property (i); there are $(T+1)^K$ such choices. Let $E$ be the number of edges $j \in \binom{[n]}{2}$ of the complete graph such that there exists $k$ such that both endpoints of $j$ lie within $S_k$. For fixed $\{S_k\}$ and $\{t_k\}$, property (i) occurs iff a certain collection of (at least) $E$ independent non-edges occur in the sampling of $\{Y^{(t)}\}$; this happens with probability at most $(1-d/n)^E \le \exp(-E d/n)$. Furthermore, we have
\begin{align*}
E \ge \binom{a_1 \Phi}{2} + \sum_{k=2}^K b_k (a_k - b_k) \Phi^2 &= \frac{a_1^2 \Phi^2}{2} - O(n) + \sum_{k=2}^K b_k (a_k - b_k) \Phi^2 \\
&= \frac{n}{d} \cdot \Phi \log d \left(\frac{a_1^2}{2} + \sum_{k=2}^K b_k(a_k-b_k) - o(1)\right),
\end{align*}
where in the first step, the first term counts edges within $S_1$ and the $k$th term of the sum counts edges within $S_k$ that have exactly one endpoint in $\cup_{\ell < k}\, S_\ell$. (Note that no edges are double-counted here.)

Putting it all together, we have
\begin{align*}
\EE[N] &\le n^{2K} (T+1)^K \sup_{\{a_k\},\{b_k\}}\exp\left\{\Phi \log d \left(a_1 + \sum_{k=2}^K b_k + o_d(1)\right)\right\} \exp\left(-\frac{d}{n} E\right) \\
\intertext{where $\{a_k\}$ and $\{b_k\}$ are subject to the constraints $a_k \ge 1+\eps$ and $b_k \in \left[\frac{\eps}{4},\frac{\eps}{2}\right]$}
&\le n^{2K} (T+1)^K \sup_{\{a_k\},\{b_k\}}\exp\left\{\Phi \log d \left(a_1 + \sum_{k=2}^K b_k - \frac{a_1^2}{2} - \sum_{k=2}^K b_k(a_k-b_k) + o(1) + o_d(1)\right)\right\} \\
&= n^{2K} (T+1)^K \sup_{\{a_k\},\{b_k\}}\exp\left\{\Phi \log d \left(a_1 - \frac{a_1^2}{2} - \sum_{k=2}^K b_k(a_k-b_k-1) + o(1) + o_d(1)\right)\right\} \\
&\le n^{2K} (T+1)^K \exp\left\{\Phi \log d \left(\frac{1}{2} - \sum_{k=2}^K \frac{\eps^2}{8} + o(1) + o_d(1)\right)\right\}
\intertext{where we have used the fact $\sup_{a \in \RR} \left(a - \frac{a^2}{2}\right) = \frac{1}{2}$ along with $a_k \ge 1+\eps$ and $b_k \in \left[\frac{\eps}{4},\frac{\eps}{2}\right]$}
&\le n^{2K} (T+1)^K \exp\left\{\Phi \log d \left(-\frac{1}{8} + o(1) + o_d(1)\right)\right\}
\intertext{where we have used $K \ge 1 + 5/\eps^2$}
&= \exp(-\Omega(n))
\end{align*}
for sufficiently large $d$. The result follows by Markov's inequality.
\end{proof}

\noindent Finally, we will need the following simple result which states that no independent set of $G(n,d/n)$ has large intersection with a fixed set of vertices.

\begin{lemma}\label{lem:spread}
Fix constants $\eps > 0$ and $a > 0$. Fix $S \subseteq [n]$ with $|S| \le a \,\frac{\log d}{d} n$. If $d > 0$ is a sufficiently large constant (depending on $\eps,a$), then with probability $1-\exp(-\Omega(n))$ there is no independent set $S'$ in $G(n,d/n)$ satisfying $|S \cap S'| \ge \eps\, \frac{\log d}{d} n$.
\end{lemma}

\begin{proof}
The proof is similar to that of Lemma~\ref{lem:max-size}. As usual, define $\Phi = \frac{\log d}{d} n$. We again use the first moment method. Let $N$ be the number of subsets $U \subseteq S$ such that $|U| = \lceil\eps\Phi\rceil =: b\Phi$ and $U$ is an independent set in $G(n,d/n)$; it is sufficient to show $N = 0$ with high probability. We have
\begin{align*}
\EE[N] &= \binom{|S|}{b\Phi} (1-d/n)^{\binom{b\Phi}{2}} \\
&\le \left(\frac{ea}{b}\right)^{b\Phi} \exp\left(-\frac{d}{n} \binom{b\Phi}{2}\right) & &\text{using \eqref{eq:binom-bound} and \eqref{eq:log-bound}} \\
&= \exp\left[b\Phi \log\left(\frac{ea}{b}\right) - \frac{db^2\Phi^2}{2n} + O(1)\right] \\
&= \exp\left[\Phi \log d \left(-\frac{b^2}{2} + o(1) + o_d(1)\right)\right] & &\text{using } b \in [\eps,a] \\
&= \exp(-\Omega(n))
\end{align*}
for sufficiently large $d$. The result follows by Markov's inequality.
\end{proof}

\subsection{Stability of Low-Degree Polynomials}

The main result of this section (Proposition~\ref{prop:stable}) states that the output of a low-degree polynomial is resilient to changes in the input, in a particular sense. Throughout this section we will use the shorthand $p := d/n$. We think of $Y \sim G(n,p)$ as simply $Y \in \{0,1\}^m$ with i.i.d.\ $\mathrm{Bernoulli}(p)$ coordinates; the graph structure will not be used in this section. We consider the hypercube graph with vertex set $\{0,1\}^m$ and an edge $(y,y')$ whenever $y,y'$ differ on exactly one coordinate.

\begin{definition}\label{def:bad}
Let $f:\{0,1\}^m\to\RR^n$ and let $c>0$. 
An edge $(y,y')$ of the hypercube $\{0,1\}^m$ is said to be \emph{$c$-bad}
for $f$ if 
\[
\|f(y)-f(y')\|^2 \geq c\, \EE_{Y \sim G(n,p)}\left[\|f(Y)\|^2\right].
\]
Also, for $y \in \{0,1\}^m$, let $B_i(y)$ denote the event that the edge traversed by flipping the $i$th coordinate of $y$ is $c$-bad for $f$.
\end{definition}

\noindent The interpolation path (Definition~\ref{def:interp}) can be thought of as a random walk on the hypercube graph (which is allowed to either remain in place or traverse an edge at each step). The following main result of this section shows that with non-trivial probability, this walk encounters no bad edges. This result is similar to Theorem~4.2 of~\cite{GJW-lowdeg} (which corresponds to the case $L=1$).

\begin{proposition}\label{prop:stable}
Let $L \in \NN$ and $c > 0$. Consider the interpolation path $Y^{(0)},\ldots,Y^{(T)}$ from Definition~\ref{def:interp} of length $T = Lm$, with $p := d/n \le 1/2$. Let $f: \{0,1\}^m \to \RR^n$ be a degree-$D$ polynomial. Then
\[ \prob\left[\text{no edge of $Y^{(0)},\ldots,Y^{(T)}$ is $c$-bad for $f$}\right] \ge p^{4LD/c}. \]
\end{proposition}

\noindent The proof will follow from the following two lemmas. The first is essentially an upper bound on the total number (weighted by the measure $G(n,p)$) of bad edges that a low-degree polynomial can have. This was proved in~\cite{GJW-lowdeg} based on standard facts about the \emph{total influence} of low-degree polynomials.

\begin{lemma}[\cite{GJW-lowdeg} Lemma 4.3]
\label{lem:total-inf}
If $p \le 1/2$ and $f: \{0,1\}^m \to \RR^n$ is a degree-$D$ polynomial then
\begin{equation}\label{eq:total-inf}
\frac{cp}{2} \sum_{i=1}^m \prob_{Y \sim G(n,p)}[B_i(Y)] \le D,
\end{equation}
where $B_i(y)$ is defined in Definition~\ref{def:bad}.
\end{lemma}

\noindent The next lemma gives an inequality that can be interpreted as follows. Roughly speaking, the right-hand side is large if there are many bad edges, and the left-hand side is large if the probability of having no bad edges on the interpolation path is small. Therefore, the inequality tells us that if the total number of bad edges is small then it is likely for the interpolation path to have no bad edges.

\begin{lemma}\label{lem:q}
Consider the interpolation path $Y^{(0)},\ldots,Y^{(T)}$ and the associated function $\sigma: [T] \to [m]$ from Definition~\ref{def:interp}.
Let $q(y)$ denote the probability that no edge of the interpolation path is bad, conditioned on the starting point $Y^{(0)} = y$. Then
\begin{equation}\label{eq:potential}
-\EE_{Y \sim G(n,p)} \log q(Y) \le H(p) \sum_{t=1}^T \prob_{Y \sim G(n,p)}[B_{\sigma(t)}(Y)]
\end{equation}
where $H$ is the binary entropy function $H(p) = -p \log p - (1-p) \log(1-p)$.
\end{lemma}

\begin{remark}
The proof of Lemma~\ref{lem:q} does not make use of the specific notion of $c$-bad from Definition~\ref{def:bad}. The result still holds if any arbitrary subset of the hypercube edges are designated ``bad'' (so long as $q(y)$ and $B_i(y)$ both use the same notion of ``bad'').
\end{remark}

\begin{remark}
Lemma~\ref{lem:q} holds not just for the specific choice of $\sigma$ from Definition~\ref{def:interp} but for any sequence $\sigma: [T] \to [m]$ of coordinates to resample. (In fact, this level of generality will be important for the inductive argument in the proof.)
\end{remark}

\begin{proof}
Proceed by induction on $T$. The base case $T = 0$ is immediate. For the case $T \ge 1$, define $\tilde q(y)$ to be the probability that the sub-walk $Y^{(1)},\ldots,Y^{(T)}$ never traverses a bad edge, conditioned on the starting point $Y^{(1)} = y$. Write $y_{-i}$ for the all-but-$i$th coordinates of $y$, and write $y_{-i}[b] \in \{0,1\}^m$ to denote the vector obtained from $y_{-i}$ by setting coordinate $i$ to the value $b \in \{0,1\}$. Note that the event $B_i(y)$ does not depend on $y_i$, so we can write $B_i(y_{-i}) := B_i(y)$. Let $j = \sigma(1)$ be the coordinate resampled in the first step. For any fixed value of $y_{-j}$, we will consider
\[ \varphi(y_{-j}) := -(1-p) \log q(y_{-j}[0]) - p \log q(y_{-j}[1]), \]
which can be thought of as the contribution from  $y_{-j}$ to the left-hand side of~\eqref{eq:potential}. If the event $B_j(y_{-j})$ holds then
\begin{align*}
\varphi(y_{-j}) &= -(1-p) \log[(1-p) \tilde q(y_{-j}[0])] - p \log[p \,\tilde q(y_{-j}[1])] \\
&= H(p) -(1-p) \log \tilde q(y_{-j}[0]) - p \log \tilde q(y_{-j}[1]),
\end{align*}
and if the complement event $\overline{B_j(y_{-j})}$ holds then
\begin{align*}
\varphi(y_{-j}) &= -\log[(1-p) \tilde q(y_{-j}[0]) + p \,\tilde q(y_{-j}[1])] \\
&\le -(1-p) \log \tilde q(y_{-j}[0]) - p \log \tilde q(y_{-j}[1])
\end{align*}
where we have used convexity of $x \mapsto -\log x$. Therefore in general we have
\[ \varphi(y_{-j}) \le H(p) \,\One_{B_j(y_{-j})} -(1-p) \log \tilde q(y_{-j}[0]) - p \log \tilde q(y_{-j}[1]). \]
Now, with $Y \sim G(n,p)$, we can write
\begin{align*}
- \EE \log q(Y) &= \EE \varphi(Y_{-j}) \\
&\le \EE[H(p) \,\One_{B_j(Y_{-j})} -(1-p) \log \tilde q(Y_{-j}[0]) - p \log \tilde q(Y_{-j}[1])] \\
&= H(p) \prob[B_j(Y)] - \EE \log \tilde q(Y).
\end{align*}
By the inductive hypothesis,
\[ -\EE \log \tilde q(Y) \le H(p) \sum_{t=2}^T \prob[B_{\sigma(t)}(Y)], \]
so this completes the proof.
\end{proof}

\begin{proof}[Proof of Proposition~\ref{prop:stable}]
We will combine Lemmas~\ref{lem:total-inf} and~\ref{lem:q}. First note that since $p \le 1/2$ we have $-p \log p \ge -(1-p) \log(1-p)$ and so
\begin{equation}\label{eq:H}
H(p) \le -2p \log p.
\end{equation}
Define $q(y)$ as in Lemma~\ref{lem:q}. The probability that no edge of the interpolation path is $c$-bad is $\EE q(Y)$ where $Y \sim G(n,p)$. We have
\begin{align*}
-\log \EE q(Y) &\le -\EE \log q(Y) & &\text{by Jensen's inequality} \\
&\le H(p) \sum_{t=1}^T \prob[B_{\sigma(t)}(Y)] & &\text{by Lemma~\ref{lem:q}} \\
&= H(p) \cdot L \sum_{i=1}^m \prob[B_i(Y)] & &\text{by Definition~\ref{def:interp}} \\
&\le H(p) \cdot L \cdot \frac{2D}{cp} & &\text{by Lemma~\ref{lem:total-inf}} \\
&\le -4c^{-1} LD \log p & &\text{by~\eqref{eq:H}}
\end{align*}
which can be rearranged to yield the result.
\end{proof}

\subsection{Putting it Together}

As in~\cite{GJW-lowdeg}, we start by observing that a random polynomial can be converted to a deterministic polynomial that works almost as well.

\begin{lemma}\label{lem:rand-det}
Suppose $f$ is a random degree-$D$ polynomial that $(k,\delta,\gamma,\eta)$-optimizes the independent set problem in $G(n,d/n)$. Then for any $c > 2$ there exists a deterministic degree-$D$ polynomial that $(k,c\delta,c\gamma,\eta)$-optimizes the independent set problem in $G(n,d/n)$.
\end{lemma}
\begin{proof}
By definition, we have $\EE_{Y,\omega} [\|f(Y,\omega)\|^2] \le \gamma k$ and $\prob_{Y,\omega}[|V_f^\eta(Y,\omega)| < k] \le \delta$. By Markov's inequality,
\[ \prob_\omega\left[\EE_Y\left[\|f(Y,\omega)\|^2\right] \ge c\gamma k\right] \le \frac{1}{c} < \frac{1}{2} \qquad \text{and} \qquad \prob_\omega\left[\prob_Y\left[|V_f^\eta(Y,\omega)| < k\right] \ge c\delta\right] \le \frac{1}{c} < \frac{1}{2} \]
and so there exists a seed $\omega^* \in \Omega$ for which the resulting deterministic polynomial $f(\cdot) = f(\cdot,\omega^*)$ satisfies
\[ \EE_Y \left[\|f(Y)\|^2\right] \le c\gamma k \qquad \text{and} \qquad \prob_Y\left[|V_f^\eta(Y)| < k\right] \le c\delta \]
as desired.
\end{proof}

\noindent We now prove our main impossibility result.

\begin{proof}[Proof of Theorem~\ref{thm:main-lower}]
For any given $\eps > 0$, set $K = \lceil 1 + 5/\eps^2 \rceil$, $T = (K-1)m$, and $\eta = \frac{\eps \log d}{16 d}$. The constant $d^* = d^*(\eps) \ge 1$ will be chosen so that $d$ is sufficiently large to apply Lemma~\ref{lem:max-size}, Proposition~\ref{prop:forbid}, and Lemma~\ref{lem:spread} in the sequel. Let $\Phi = \frac{\log d}{d} n$.

Assume on the contrary that the random polynomial that we wish to rule out, exists. By Lemma~\ref{lem:rand-det}, there exists a deterministic degree-$D$ polynomial $f$ that satisfies
\begin{equation}\label{eq:f-det-prop}
\EE_Y \left[\|f(Y)\|^2\right] \le 3\gamma(1+\eps)\Phi \qquad \text{and} \qquad \prob_Y\left[|V_f^\eta(Y)| < (1+\eps)\Phi\right] \le 3\delta.
\end{equation}

Sample the interpolation path $Y^{(0)},\ldots,Y^{(T)}$ as in Definition~\ref{def:interp}, and let $U_t = V_f^\eta(Y^{(t)})$ be the resulting independent sets. Consider the following process to construct a sequence of sets $S_1, \ldots, S_K \subseteq [n]$. Let $S_1 = U_0$. Then for $k = 2,3,\ldots,K$, let $S_k = U_{t_k}$ where $t_k \in [T]$ is the minimum $t$ for which $|U_t \setminus (\cup_{\ell < k}\, S_\ell)| \ge \frac{\eps}{4} \Phi$; if no such $t$ exists then the process fails. We will show that with positive probability (over the interpolation path), the following events all occur simultaneously:
\begin{enumerate}
    \item[(i)] $|U_t| \ge (1+\eps)\Phi$ for all $0 \le t \le T$, and the process $S_1,\ldots,S_K$ succeeds,
    \item[(ii)] no edge on the interpolation path is $c$-bad for $f$, where $c = \frac{\eps}{96 \gamma (1+\eps)}$,
    \item[(iii)] the conclusion of Proposition~\ref{prop:forbid} holds (i.e., no instance of the forbidden structure exists).
\end{enumerate}

\noindent We will first show that events (i)-(iii) occur simultaneously with positive probability, and then we will show that this yields a contradiction. By Proposition~\ref{prop:stable}, event (ii) occurs with probability at least $(d/n)^{4(K-1)D/c}$. By Proposition~\ref{prop:forbid}, event (iii) occurs with probability $1-\exp(-\Omega(n))$. It remains to consider event (i).

For each fixed $t$ we have that $Y^{(t)}$ is distributed as $G(n,d/n)$, so by combining Lemma~\ref{lem:max-size} with the second property of $f$ from~\eqref{eq:f-det-prop}, we have $(1+\eps)\Phi \le |U_t| \le (2+\eps)\Phi$ with probability at least $1 - 3\delta - \exp(-\Omega(n))$; we will take a union bound over $t$. Now suppose that for some $0 \le T' \le T-m$, $Y^{(0)},\ldots,Y^{(T')}$ have been sampled so far, and $0 = t_1 < t_2 < \cdots < t_{K'}$ are the indices of the sets $S_k = U_{t_k}$ selected so far ($t_{K'} \le T'$). Note that $Y^{(T'+m)}$ is independent from $\{Y^{(t)}\}_{t \le T '}$ and so, provided $|S_k| \le (2+\eps)\Phi$ for $1 \le k \le K'$, Lemma~\ref{lem:spread} (with $S = \cup_{k \le K'}\, S_k$ and $a = (2+\eps)K'$) implies $|U_{T'+m} \cap (\cup_{k \le K'}\, S_k)| \le \eps \Phi$ with probability $1-\exp(-\Omega(n))$. Provided $|U_{T'+m}| \ge (1+\eps)\Phi$, this implies $|U_{T'+m} \setminus (\cup_{k \le K'}\, S_k)| \ge \Phi \ge \frac{\eps}{4}\Phi$ and so $t_{K'+1} \le T'+m$; thus by induction, $t_k \le (k-1)m$ for all $k \in [K]$ and so the process $\{S_k\}$ succeeds by timestep $T = (K-1)m$. We therefore conclude that event (i) holds with probability at least $1 - 3\delta(T+1) - \exp(-\Omega(n))$.

Using $3\delta(T+1) = 3\delta[(K-1)m+1] \le 3\delta Km$, we now have that events (i)-(iii) occur simultaneously with positive probability, provided
\begin{equation}\label{eq:pos-prob}
(d/n)^{4(K-1)D/c} > 3\delta Km + \exp(-\Omega(n)).
\end{equation}
For sufficiently large $n$, the term $\exp(-\Omega(n))$ is at most $\exp(-Cn)$ for some constant $C = C(\eps,d) > 0$. Also recall $m = \binom{n}{2} < \frac{n^2}{2}$. Thus, to satisfy~\eqref{eq:pos-prob}, it is sufficient to have
\begin{equation}\label{eq:2-cond}
(d/n)^{4(K-1)D/c} \ge 3\delta K n^2 \qquad \text{and} \qquad (d/n)^{4(K-1)D/c} \ge 2 \exp(-Cn).
\end{equation}
For $d \ge 1$, the second condition in~\eqref{eq:2-cond} is implied by $D \le (Cn - \log 2) \frac{c}{4(K-1)\log n}$. For sufficiently large $n$, and using $c = \frac{\eps}{96\gamma(1+\eps)}$, this is implied by $D \le \frac{C_1 n}{\gamma \log n}$, where $C_1 = C_1(\eps,d) > 0$ is a constant. For $d \ge 1$, the first condition in~\eqref{eq:2-cond} is implied by $\delta \le \exp[-\frac{4(K-1)D}{c} \log n - 2 \log n - \log(3K)]$. Since $\gamma \ge 1$ and $D \ge 1$, for sufficiently large $n$ this is implied by $\delta \le \exp(-C_2 \gamma D \log n)$ for another constant $C_2 = C_2(\eps,d) > 0$.

To complete the proof, it remains to show that if events (i)-(iii) occur simultaneously, this results in a contradiction. The idea is to use the stability property from (ii) to show that the sets $S_1,\ldots,S_K$ from (i) are an instance of the forbidden structure that is disallowed by (iii).

We will first show $|U_t \symd U_{t-1}| \le \frac{\eps}{4} \Phi$ for all $1 \le t \le T$, where $\symd$ denotes symmetric difference. From (i) we know that the failure event in $V^\eta_f$ (the second case of~\eqref{eq:V}) does not occur on any of the inputs $Y^{(t)}$. Therefore, the definition of $V_f^\eta$ (Definition~\ref{def:V}) implies that there are at least $|U_t \symd U_{t-1}| - 2 \eta n$ coordinates $i \in [n]$ for which $|f_i(Y^{(t)}) - f_i(Y^{(t-1)})| \ge 1/2$. To see this, note that $\One_{i \in U_t}$ can only differ from $\One_{i \in U_{t-1}}$ if either
\begin{itemize}
\item $i$ lies in the set $(A \setminus \tilde{A}) \cup B$ (see Definition~\ref{def:V}) for either $V_f^\eta(Y^{(t)})$ or $V_f^\eta(Y^{(t-1)})$, or
\item among the values $f_i(Y^{(t)})$ and $f_i(Y^{(t-1)})$, one is $\ge 1$ and the other is $\le 1/2$.
\end{itemize}
This means
\[ \frac{1}{4}(|U_t \symd U_{t-1}| - 2 \eta n) \le \|f(Y^{(t)}) - f(Y^{(t-1)})\|^2 \le c \EE_{Y \sim G(n,d/n)} \left[\|f(Y)\|^2\right] \le 3c \gamma (1+\eps)\Phi \]
where we have used event (ii) along with the definition of $c$-bad (Definition~\ref{def:bad}) and the first property of $f$ from~\eqref{eq:f-det-prop}. Rearranging this yields
\[ |U_t \symd U_{t-1}| \le 12c \gamma (1+\eps)\Phi + 2 \eta n = \frac{\eps}{4}\Phi \]
as desired, where we have used $c = \frac{\eps}{96 \gamma (1+\eps)}$ and $\eta = \frac{\eps \log d}{16 d}$.

Recall that $S_k$ is the first $U_t$ for which $|U_t \setminus (\cup_{\ell < k}\, S_\ell)| \ge \frac{\eps}{4}\Phi$. Using the fact $|U_t \symd U_{t-1}| \le \frac{\eps}{4}\Phi$ from above, this means $|S_k \setminus (\cup_{\ell < k}\, S_\ell)| \le \frac{\eps}{2}\Phi$. Combining this with event (i) and the fact that $S_k$ is an independent set in $Y^{(t_k)}$, we have that $S_1,\ldots,S_K$ satisfies the properties of the forbidden structure from event (iii). This yields a contradiction and completes the proof.
\end{proof}

\section{Proof of Achievability}\label{sec:upper}

In this section we prove our main achievability result (Theorem~\ref{thm:main-upper}) which shows that low-degree algorithms can find independent sets of size $(1-\eps) \frac{\log d}{d} n$. We begin by defining some terminology pertaining to local algorithms on graphs. Throughout this section we will consider graphs $G = (V,E)$ with possibly-infinite vertex set $V$, but which are \emph{locally finite}, i.e., each vertex has a finite number of neighbors. We will consider functions that take as input $(G,v)$ where $G = (V,E)$ is a graph and $v \in V$ is a designated ``root'' vertex; let $\Lambda$ denote the set of such $(G,v)$ pairs. We will also consider functions that take as input $(G,v,X)$ where $G$ and $v$ are as before and $X: V \to [0,1]$ is a labelling of the vertices; let $\tilde{\Lambda}$ denote the set of such $(G,v,X)$ pairs.

For a graph $G = (V,E)$ and a vertex $v \in V$, the \emph{$r$-neighborhood} of $v$, denoted $N_r(G,v)$, is the rooted graph with root $v$ that contains all vertices reachable from $v$ by a path of length $\le r$, along with all edges on such paths. We will use $|N_r(G,v)|$ to denote the number of edges in the $r$-neighborhood. Two rooted graphs are said to be isomorphic if there is a root-preserving graph isomorphism between them. A function $g$ with domain $\Lambda$ is said to be $r$-local if $g(G,v)$ depends only on the isomorphism class of $N_r(G,v)$. (Informally, $g$ has access to the ``shape'' of the $r$-neighborhood but not the identity of the specific vertices.)

In the presence of vertex labels $X: V \to [0,1]$, we generalize the above notions as follows. The \emph{labeled $r$-neighborhood} of $v$ in $G$, denoted $\tilde{N}_r(G,v,X)$, is the $r$-neighborhood along with the vertex labels given by $X$ (restricted to the $r$-neighborhood). Two rooted labeled graphs are said to be isomorphic if there is a root-preserving and label-preserving graph isomorphism between them. A function $h$ with domain $\tilde\Lambda$ is said to be $r$-local if $h(G,v,X)$ depends only on the isomorphism class of $\tilde{N}_r(G,v,X)$.

The \emph{Poisson Galton--Watson tree} with parameter $d > 0$, denoted $\PGW(d)$, is the distribution over rooted (possibly-infinite) trees $(T,o)$ generated as follows:
\begin{itemize}
    \item Start with a root vertex $o$ at level $0$.
    \item For $\ell = 0,1,2,\ldots$, each vertex at level $\ell$ independently spawns $\Pois(d)$ child vertices at level $\ell+1$.
    \item Every vertex (except the root) is connected to its parent by an edge.
\end{itemize}

\noindent It is well-known that the distribution of the $r$-neighborhood of any fixed vertex in $G(n,d/n)$ converges to the $r$-neighborhood of the root in $\PGW(d)$ as $n \to \infty$ with $r$ held fixed (as discussed in e.g.~\cite{half-opt}); see Lemma~\ref{lem:graph-tree} below for one precise sense in which this convergence holds.

An \emph{$r$-local algorithm} for the maximum independent set problem is an $r$-local function $h: \tilde\Lambda \to \{0,1\}$ with the property that $\{v \in V \,:\, h(G,v,X) = 1\}$ is an independent set for any graph $G = (V,E)$ with any vertex labels $X$. A line of prior work~\cite{LW-indep,HLS,GS-local,half-opt} has considered the problem of choosing $h$ to maximize the expected size of the independent set when $G \sim G(n,d/n)$ and $X$ is i.i.d.\ $\Unif([0,1])$. Due to the convergence of local neighborhoods to $\PGW(d)$, this task is equivalent (up to sub-leading terms in $n$) 
to maximizing the probability that $h(T,o,X) = 1$ when $(T,o) \sim \PGW(d)$ and $X$ is again i.i.d.\ $\Unif([0,1])$.

The following result of~\cite{half-opt} shows that local algorithms can produce large independent sets in $\PGW(d)$. As discussed in Section~4 of~\cite{half-opt}, this implies that local algorithms can produce independent sets of expected size $(1-\eps) \frac{\log d}{d} n$ in $G(n,d/n)$.

\begin{theorem}[\cite{half-opt}~Theorem~4.1]
\label{thm:local-alg}
For any $\eps > 0$ and any sufficiently large $d$ (depending on $\eps$), there exists $r = r(\eps,d)$ and an $r$-local function $h: \tilde\Lambda \to \{0,1\}$ satisfying the following. If $(T,o) \sim \PGW(d)$ and vertex labels $\{X_v\}_{v \in V(T)}$ are drawn i.i.d.\ from the uniform distribution on $[0,1]$, then
\begin{itemize}
\item the vertex set $\{v \in V(T) \,:\, h(T,o,X) = 1\}$ is an independent set in $T$ with probability 1, and
\item $\EE[h(T,o,X)] \ge (1-\eps) \frac{\log d}{d}$.
\end{itemize}
\end{theorem}

\begin{remark}\label{rem:local-to-ld}
Our proof of Theorem~\ref{thm:main-upper} will show how to approximate the local algorithm from Theorem~\ref{thm:local-alg} by a low-degree algorithm. We will not use any specifics of the local algorithm, and so our proof actually shows how to approximate \emph{any} local algorithm by a low-degree algorithm. More precisely: for any fixed $\eps > 0$, $\eta > 0$, and $d \ge 1$, if we are given an $r$-local algorithm $h$ for independent sets with $\EE[h(T,o,X)] \ge \alpha$, then for any $n \ge n^*(\eps,\eta,d,r,\alpha)$ we can produce a deterministic degree-$D$ polynomial that $(k,\delta,\gamma,\eta)$-optimizes the independent set problem in $G(n,d/n)$ with parameters $k = (1-\eps)\alpha n$ and $\delta = \exp(-C n^{1/3})$ where $D > 0, \gamma \ge 1, C > 0$ are constants depending on $\eps,\eta,d,r,\alpha$.
\end{remark}

\noindent The next result, which is a special case of Lemma~12.4 of~\cite{very-sparse}, quantifies the convergence of local neighborhoods of $G(n,d/n)$ to $\PGW(d)$.

\begin{lemma}[see \cite{very-sparse} Lemma~12.4]
\label{lem:graph-tree}
Let $G \sim G(n,d/n)$, and let $(T,o) \sim \PGW(d)$. Let $g: \Lambda \to [-1,1]$ be an $r$-local function. For all sufficiently large $n$ (depending on $d,r$) and for any $v \in [n]$,
\[ \left|\EE[g(G,v)] - \EE[g(T,o)]\right| \le c n^{-1/4} \log n \]
where $c > 0$ is a universal constant.
\end{lemma}

\noindent The next result is a special case of (the first statement in) Proposition~12.3 of~\cite{very-sparse}.

\begin{proposition}[see~\cite{very-sparse} Proposition~12.3]
\label{prop:local-moment}
Let $G \sim G(n,d/n)$ with $d \ge 1$. Let $g: \Lambda \to [-1,1]$ be an $r$-local function. For all $p \ge 2$,
\[ \EE\left[\,\left|\sum_{v \in [n]} g(G,v) - \EE \sum_{v \in [n]} g(G,v)\right|^p\,\right] \le \left(c \sqrt{n} p^{3/2} (2d)^r \right)^p \]
where $c > 0$ is a universal constant.
\end{proposition}

\noindent A simple consequence of the above moment inequality is a tail bound for local functions.

\begin{corollary}\label{cor:local-conc}
Let $G \sim G(n,d/n)$ with $d \ge 1$. Let $g: \Lambda \to [-1,1]$ be an $r$-local function. For a universal constant $c > 0$ and for all $t \ge (2e)^{3/2} c \sqrt{n} (2d)^r$,
\[ \prob\left[\,\left|\sum_{v \in [n]} g(G,v) - \EE \sum_{v \in [n]} g(G,v)\right| \ge t\right] \le \exp\left(-\frac{3t^{2/3}}{2ec^{2/3}n^{1/3}(2d)^{2r/3}}\right). \]
\end{corollary}
\begin{proof}
Let $c$ be the constant from Proposition~\ref{prop:local-moment}. Choosing $p = e^{-1} [c \sqrt{n} (2d)^r / t]^{-2/3} \ge 2$,
\begin{align*}
\prob\left[\,\left|\sum_{v \in [n]} g(G,v) - \EE \sum_{v \in [n]} g(G,v)\right| \ge t\right] &= \prob\left[\,\left|\sum_{v \in [n]} g(G,v) - \EE \sum_{v \in [n]} g(G,v)\right|^p \ge t^p\right] \\
&\le t^{-p} \EE\left[\,\left|\sum_{v \in [n]} g(G,v) - \EE \sum_{v \in [n]} g(G,v)\right|^p\,\right] \\
&\le t^{-p} \left(c \sqrt{n} p^{3/2} (2d)^r \right)^p \\
&= \exp\left(-\frac{3}{2}\,p\right) \\
&= \exp\left(-\frac{3t^{2/3}}{2ec^{2/3}n^{1/3}(2d)^{2r/3}}\right)
\end{align*}
as desired.
\end{proof}

\noindent We will also need the following standard multiplicative version of the Chernoff bound~\cite{UM-book}.

\begin{proposition}\label{prop:chernoff}
Suppose $Z_1,\ldots,Z_n$ are independent, taking values in $\{0,1\}$. Let $Z = \sum_i Z_i$ and $\mu = \EE[Z]$. For any $0 \le \delta \le 1$,
\[ \prob[Z \le (1-\delta)\mu] \le \exp\left(-\frac{\delta^2 \mu}{2}\right). \]
Also, for any $\delta \ge 0$,
\[ \prob[Z \ge (1+\delta)\mu] \le \exp\left(-\frac{\delta^2 \mu}{2+\delta}\right), \]
and so for $\delta \ge 1$,
\[ \prob[Z \ge (1+\delta)\mu] \le \exp\left(-\frac{\delta \mu}{3}\right). \]
\end{proposition}

\begin{proof}[Proof of Theorem~\ref{thm:main-upper}]

Given $\eps > 0$, apply Theorem~\ref{thm:local-alg} (with $\eps/5$ in place of $\eps$) to obtain $d^*(\eps) > 1$, $r = r(\eps,d)$ and an $r$-local function $h: \tilde\Lambda \to \{0,1\}$ that outputs independent sets with $\EE[h(T,o,X)] \ge (1-\eps/5) \frac{\log d}{d}$ when $(T,o) \sim \PGW(d)$ and $X$ is i.i.d.\ $\Unif([0,1])$.

By Lemma~\ref{lem:rand-det}, it is sufficient to prove the result for a \emph{random} polynomial instead of a deterministic one (up to a change in the constants $\gamma, C$). We will construct a random polynomial $f: \{0,1\}^{\binom{n}{2}} \to \RR^n$ as follows. The input $Y$ to $f$ encodes a graph on vertex set $[n]$. The internal randomness of $f$ samples vertex labels $ \{X_v\}_{v \in [n]}$ i.i.d.\ from $\Unif([0,1])$. We will construct $f$ with the following property:
\begin{equation}\label{eq:f-prop}
\text{for any $v \in [n]$, if $N_r(Y,v)$ is a tree with $|N_r(Y,v)| \le s$ then $f_v(Y,X) = h(Y,v,X)$}
\end{equation}
where $s = s(\eps,d,\eta)$ is a constant to be chosen later.

Concretely, we construct $f$ as follows. Let $\mathcal{G}_{v,r,s}$ be the collection of graphs $G$ on vertex set $[n]$ for which $|E(G)| \le s$ and every non-isolated vertex is reachable from $v$ by a path of length $\le r$. (In other words, $\mathcal{G}_{v,r,s}$ consists of all possible $r$-neighborhoods for $v$ of size $\le s$.) Let
\begin{equation}\label{eq:f-expansion}
f_v(Y,X) = \sum_{G \in \mathcal{G}_{v,r,s}} \alpha(G,v,X) \prod_{e \in E(G)} Y_e
\end{equation}
where the coefficients $\alpha(G,v,X)$ are chosen so that~\eqref{eq:f-prop} is satisfied, i.e., $\alpha(G,v,X)$ are defined recursively by
\begin{equation}\label{eq:coeff}
\alpha(G,v,X) = h(G,v,X) - \sum_{\substack{G' \in \mathcal{G}_{v,r,s} \\ E(G') \subsetneq E(G)}} \alpha(G',v,X).
\end{equation}

Let $\sT$ be a set of rooted trees consisting of one representative from each isomorphism class of rooted trees of depth at most $2r$. Let $\sT_s \subseteq \sT$ contain only those trees with at most $s$ edges. Let $Y \sim G(n,d/n)$, and for $T \in \sT$, let $n_T$ denote the number of occurrences of the neighborhood $T$ in $Y$, i.e.,
\[ n_T = |\{v \in [n] \,:\, N_{2r}(Y,v) \cong T\}| \]
where $\cong$ denotes isomorphism of rooted graphs. Also, for $T \in \sT$, let $p_T$ denote the probability that $T$ occurrs as the neighborhood of the root in $\PGW(d)$, i.e.,
\[ p_T = \prob_{(U,o) \sim \PGW(d)}[N_{2r}(U,o) \cong T]. \]
Also, let $\phi_T$ denote the probability over $X$ that $h(Y,v,X) = 1$ conditioned on $N_{2r}(Y,v) \cong T$. (Note that the event $\{h(Y,v,X) = 1\}$ depends only on $X$ and $N_r(Y,v)$ since $h$ is $r$-local.)

By applying Lemma~\ref{lem:graph-tree} to the function $g(G,v) = \One_{N_{2r}(G,v) \cong T}$, we have
\begin{equation}\label{eq:n_T-mean}
\left|\EE[n_T] - p_T n \right| \le cn^{3/4} \log n \end{equation}
for sufficiently large $n$ (depending on $d,r$). By Corollary~\ref{cor:local-conc}, for any $t \ge (2e)^{3/2} c \sqrt{n} (2d)^{2r}$,
\begin{equation}\label{eq:n_T-conc}
\prob\left[\left|n_T - \EE[n_T]\right| \ge t\right] \le \exp\left(-\frac{3t^{2/3}}{2ec^{2/3}n^{1/3}(2d)^{2r/3}}\right).
\end{equation}
Combining~\eqref{eq:n_T-mean} and~\eqref{eq:n_T-conc} we have the following: for any $\tau > 0$ and for sufficiently large $n$ (depending on $d,r,\tau$),
\begin{equation}\label{eq:n_T-comb}
\prob\left[\left|n_T - p_T n\right| \ge \tau n\right] \le \exp(-Cn^{1/3})
\end{equation}
for some $C = C(d,r,\tau) > 0$.

We will show that with high probability, the rounding procedure $V_f^\eta(Y,X)$ does not encounter the failure event (the second case of~\eqref{eq:V}). Suppose some vertex $v$ is such that $N_{2r}(Y,v)$ is a tree with $|N_{2r}(Y,v)| \le s$. Then for all $u \in N_1(Y,v)$ we have that $N_r(Y,u)$ is a tree with $|N_r(Y,u)| \le s$ and so by~\eqref{eq:f-prop}, $f_u(Y,X) = h(Y,u,X) \in \{0,1\}$. Since $h$ outputs independent sets, it follows that $v$ is not in the ``bad'' set $(A \setminus \tilde A) \cup B$ from the definition of $V_f^\eta$ (Definition~\ref{def:V}). We have now shown that $(A \setminus \tilde A) \cup B$ is disjoint from the set
\[ V_s := \bigcup_{T \in \sT_s} \{v \in [n] \,:\, N_{2r}(Y,v) \cong T\}. \]
For each $T \in \sT_s$, we have from~\eqref{eq:n_T-comb} that $|n_T - p_T n| \le \eta n / (2|\sT_s|)$ with probability $1-\exp(-\Omega(n^{1/3}))$ where $\Omega(\cdot)$ hides a constant depending on $\eps,d,r,s,\eta$. Choose $s$ large enough so that $\sum_{T \in \sT_s} p_T \ge 1 - \eta/2$. We now have
\[ |A \setminus \tilde A| + |B|
= |(A \setminus \tilde A) \cup B|
\le n - |V_s|
= n - \sum_{T \in \sT_s} n_T
\le n - \sum_{T \in \sT_s} \left(p_T n - \frac{\eta n}{2|\sT_s|}\right)
= \left(1 - \sum_{T \in \sT_s} p_T\right) n + \frac{\eta n}{2}
\le \eta n. \]
In conclusion, $V_f^\eta(Y,X)$ avoids the failure event with probability $1 - \exp(-\Omega(n^{1/3}))$.

Next we will show that the independent set $I := V_f^\eta(Y,X)$ is large with high probability. From the guarantees on $h$,
\[ \left(1-\frac{\eps}{5}\right) \frac{\log d}{d} \le \EE_{(T,o) \sim \PGW(d)}[h(T,o,X)] = \sum_{T \in \sT} p_T \phi_T. \]
Choose $s$ large enough so that $\sum_{T \in \sT_s} p_T \ge 1 - \frac{\eps}{5} \frac{\log d}{d}$. Since $\phi_T \in [0,1]$, this implies
\[ \sum_{T \in \sT_s} p_T \phi_T \ge \left(\sum_{T \in \sT} p_T \phi_T\right) - \frac{\eps}{5} \frac{\log d}{d} \ge \left(1 - \frac{2\eps}{5}\right) \frac{\log d}{d}. \]
Again using~\eqref{eq:n_T-comb}, with probability $1 - \exp(-\Omega(n^{1/3}))$ over $Y$,
\begin{equation}\label{eq:nphi}
\sum_{T \in \sT_s} n_T \phi_T
\ge \sum_{T \in \sT_s} \left(p_T n - \frac{\eps}{5 |\sT_s|} \frac{\log d}{d} n\right) \phi_T
\ge \left(\sum_{T \in \sT_s} p_T \phi_T\right) n - \frac{\eps}{5} \frac{\log d}{d} n \ge \left(1 - \frac{3\eps}{5}\right) \frac{\log d}{d} n.
\end{equation}

\noindent Now fix $Y$ satisfying~\eqref{eq:nphi} and consider the randomness of $X$. Recall from above that $(A \setminus \tilde A) \cup B$ is disjoint from $V_s$. Thus, if $v$ satisfies $N_{2r}(Y,v) \cong T$ for some $T \in \sT_s$ then $v$ will be included in the independent set $I := V_f^\eta(Y,X)$ iff $h(Y,v,X) = 1$, which occurs with probability $\phi_v := \phi_T$ (over the randomness of $X$). We will partition the elements of $V_s$ into ``bins'' $W_1,\ldots,W_{s+1}$ such that for each bin $W_i$, the vertices in $W_i$ have disjoint $r$-neighborhoods and so the random variables $\{\One_{v \in I}\}_{v \in W_i}$ are independent (conditioned on $Y$). Each vertex $v \in V_s$ has at most $s+1$ vertices in its $2r$-neighborhood, and so there are at most $s$ vertices $u \in V_s$ such that $u \ne v$ and $N_r(Y,v) \cap N_r(Y,u) \ne \emptyset$. Since there are $s+1$ bins, we can greedily assign vertices to bins in order to achieve the desired disjointness property. Now that the bins $\{W_i\}$ have been constructed, we have by the Chernoff bound (Proposition~\ref{prop:chernoff}) that for each $i$,
\begin{equation}\label{eq:W-chernoff}
\prob_X\left[\sum_{v \in W_i} \One_{v \in I} \le \left(1 - \frac{\eps}{5}\right)\mu_i\right] \le \exp\left(-\frac{1}{2}\left(\frac{\eps}{5}\right)^2 \mu_i \right)
\end{equation}
where
\begin{equation}\label{eq:mu-defn}
\mu_i = \EE_X \sum_{v \in W_i} \One_{v \in I} = \sum_{v \in W_i} \phi_v.
\end{equation}
Call a bin $W_i$ ``large'' if $\mu_i \ge \frac{\eps}{5(s+1)} \frac{\log d}{d} n$ and ``small'' otherwise. Using~\eqref{eq:W-chernoff} and a union bound over $i$, we have with probability $1-\exp(-\Omega(n))$ that every large bin $W_i$ satisfies
$\sum_{v \in W_i} \One_{v \in I} \ge \left(1 - \frac{\eps}{5}\right)\mu_i$. Provided this holds, we now have
\allowdisplaybreaks
\begin{align*}
|I| &\ge \sum_{v \in V_s} \One_{v \in I} \\
&\ge \sum_{i \,:\, W_i\text{ large}} \;\sum_{v \in W_i} \One_{v \in I} \\
&\ge \sum_{i \,:\, W_i\text{ large}} \left(1-\frac{\eps}{5}\right) \mu_i \\
&= \left(1-\frac{\eps}{5}\right) \left[\sum_i \mu_i - \sum_{i \,:\, W_i\text{ small}} \mu_i \right] \\
&\ge \left(1-\frac{\eps}{5}\right) \left[\left(\sum_i \mu_i\right) - \frac{\eps}{5} \frac{\log d}{d} n \right]& &\text{using the definition of ``small''} \\
&= \left(1-\frac{\eps}{5}\right) \left[\left(\sum_{v \in V_s} \phi_v\right) - \frac{\eps}{5} \frac{\log d}{d} n \right]& &\text{using the definition of $\mu_i$~\eqref{eq:mu-defn}} \\
&= \left(1-\frac{\eps}{5}\right) \left[\left(\sum_{T \in \sT_s} n_T \phi_T\right) - \frac{\eps}{5} \frac{\log d}{d} n \right] \\
&\ge \left(1-\frac{\eps}{5}\right) \left[\left(1 - \frac{3\eps}{5}\right)\frac{\log d}{d} n - \frac{\eps}{5} \frac{\log d}{d} n \right]& &\text{using~\eqref{eq:nphi}} \\
&= \left(1-\frac{\eps}{5}\right) \left(1 - \frac{4\eps}{5}\right)\frac{\log d}{d} n \\
&\ge (1-\eps) \frac{\log d}{d} n.
\end{align*}
Therefore, the independent set has size $|I| \ge (1-\eps)\frac{\log d}{d} n$ with probability $1-\exp(-\Omega(n^{1/3}))$ over both $Y$ and $X$.

Finally, we need to check the normalization condition: $\EE_{Y,X}[\|f(Y,X)\|^2] \le \gamma (1-\eps)\frac{\log d}{d}n$ for a constant $\gamma = \gamma(\eps,d,\eta) \ge 1$. By linearity of expectation, it is sufficient to show $\EE_{Y,X}[f_v(Y,X)^2] = O(1)$ uniformly over $v$. Fix a vertex $v \in [n]$ and define the random variable $N = |N_r(Y,v)|$. Recall the expansion~\eqref{eq:f-expansion} for $f_v$. For each $G \in \mathcal{G}_{v,r,s}$, the corresponding term in the sum can be nonzero only if $G$ is a subgraph of $N_r(Y,v)$. Thus, the number of nonzero terms is at most
\[ \binom{N}{\le s} = \sum_{i=0}^s \binom{N}{i} \le \sum_{i=0}^s N^i \le (N+1)^s. \]
Furthermore, we can see from~\eqref{eq:coeff} that the coefficient of each term is bounded by a constant, uniformly over $v$ and $X$: $|\alpha(G,v,X)| \le a$ for some $a = a(r,s)$. This means
\begin{equation}\label{eq:f2}
f_v(Y,X)^2 \le [a(N+1)^s]^2 = a^2(N+1)^{2s}.
\end{equation}

\noindent In order to bound the expectation of this quantity, we will need a tail bound for $N$. Starting from $m_0 = 1$, let $m_i$ be the number of vertices whose distance in $Y$ from $v$ is exactly $i$. Conditioned on $m_i$, we have that $m_{i+1}$ is stochastically dominated by Binomial$(m_i n, d/n)$. Using the Chernoff bound (Proposition~\ref{prop:chernoff}), for fixed $m_i \ge 1$ and any $\delta \ge 1$,
\[ \prob[m_{i+1} \ge (1+\delta)dm_i] \le \exp\left(-\frac{\delta d m_i}{3}\right) \le \exp\left(-\frac{\delta d}{3}\right). \]
Therefore, with probability at least $1 - r \exp(-\delta d/3)$, we have $m_i < [(1+\delta)d]^i$ for all $0 \le i \le r$ and so
\[ N < \sum_{i=0}^r [(1+\delta)d]^i \le [(1+\delta)d+1]^r. \]
For $\delta \ge 1$ and $d \ge 1$ we have $(1+\delta)d + 1 \le 2\delta d + 1 \le 3\delta d$, so we can rewrite the above as
\[ \prob[N \ge (3\delta d)^r] \le r \exp(-\delta d/3). \]
Letting $t = (3\delta d)^r$, we now have a tail bound for $N$: for all  $t \ge (3d)^r$,
\[ \prob[N \ge t] \le r \exp(-t^{1/r}/9). \]
Finally, combining this with~\eqref{eq:f2}, we have
\[ \EE_{Y,X}[f_v(Y,X)^2]
\le \sum_{t=0}^\infty a^2 (t+1)^{2s} \prob[N=t]
\le \sum_{t=0}^{\lceil(3d)^r\rceil} a^2 (t+1)^{2s} + \sum_{t=\lceil(3d)^r\rceil}^\infty a^2 (t+1)^{2s} r \exp(-t^{1/r}/9), \]
which is finite and independent of $n$. This completes the proof.
\end{proof}

\section*{Acknowledgments}

The author is grateful to Charles Bordenave for helpful discussions regarding concentration of local functions on graphs (Proposition~\ref{prop:local-moment} and Corollary~\ref{cor:local-conc}).

\bibliographystyle{alpha}
\bibliography{main}

\end{document}